\documentclass[journal]{IEEEtran}
\usepackage{amsmath,amsfonts}
\usepackage{algorithmic}
\usepackage{algorithm}
\usepackage{array}
\usepackage[caption=false,font=normalsize,labelfont=sf,textfont=sf]{subfig}
\usepackage{textcomp}
\usepackage{stfloats}
\usepackage{url}
\usepackage{diagbox}
\usepackage{verbatim}
\usepackage{graphicx}
\usepackage{multirow}
\usepackage{cite}
\usepackage{algorithmic}
\usepackage{amsmath}
\usepackage{amsthm}
\newtheorem{remark}{Remark}
\newtheorem{proposition}{Proposition}
\usepackage{threeparttable}
\usepackage{xcolor}
\begin{document}
\title{Long-Term Carbon-Efficient Planning for Geographically Shiftable Resources: A Monte Carlo Tree Search Approach}

\author{Xuan He,~\IEEEmembership{Graduate Student Member, IEEE}, Danny H.K. Tsang,~\IEEEmembership{Life Fellow, IEEE} and Yize  Chen,~\IEEEmembership{Member, IEEE} \vspace{-10pt}
\thanks{X. He and D.H.K. Tsang are with the Information Hub, Hong Kong University of Science and Technology (Guangzhou), email: xhe085@connect.hkust-gz.edu.cn, eetsang@ust.hk. Y. Chen is with the Department of Electrical and Computer Engineering, University of Alberta, email: yize.chen@ualberta.ca. 

This work was supported in part by Guangzhou Municipal Science and Technology Project under Grant 2023A03J0011, in part by Guangdong Provincial Key Laboratory of Integrated Communications, Sensing and Computation for Ubiquitous Internet of Things, and in part by the National Foreign Expert Project, Project Number G2022030026L.}
}


\markboth{
IEEE Transactions on Power Systems}%
{Shell \MakeLowercase{\textit{et al.}}: 
}


\maketitle

\begin{abstract} The global climate challenge is demanding urgent actions for decarbonization, while electric power systems take the major roles in the clean energy transition. Due to the existence of spatially and temporally dispersed renewable energy resources and the uneven distribution of carbon emission intensity throughout the grid, it is worth investigating future load planning and demand management to offset those generations with higher carbon emission rates. Such techniques include inter-region utilization of geographically shiftable resources and stochastic renewable energy. For instance, data centers are perceived as major carbon emission producers in the future due to increasing information load, while they hold untapped capability of geographical load balancing.
In this paper, we focus on locating and operating geographically shiftable resources, and propose a novel planning and operation model minimizing the system-level carbon emissions. This model decides the optimal locations for shiftable resource expansion along with the power dispatch schedule. To accommodate future system operation patterns and a wide range of operating conditions, we incorporate 20-year fine-grained load and renewables scenarios for grid simulations of realistic sizes (e.g., up to 1888 buses). To tackle the computational challenges coming from the combinatorial nature of such large-scale planning problems,  we develop a customized Monte Carlo Tree Search (MCTS) method, which can find reasonable solutions satisfying solution time limits. Besides, MCTS enables flexible time window settings and offline solution adjustments. Extensive simulations validate that our planning model can reduce more than 10\% carbon emission across all setups. Compared to off-the-shelf optimization solvers such as Gurobi, our method achieves up to 8.1X acceleration, while the solution gaps are less than 1.5\% in large-scale cases. Moreover, our method can find reasonable solutions within different time windows, unlike other benchmarked solvers that require considerable time or fail to produce feasible solutions.

\end{abstract}


\section{Introduction}
\IEEEPARstart{G}lobal carbon emissions continue to increase, while limited time is left for meeting our global climate targets~\cite{liu2023monitoring}. As one of the biggest contributor to greenhouse gas emissions, electric power systems appeal for a carbon-free and sustainable transition. This motivates various ongoing and future projects to take carbon emission reduction as a top priority \cite{masanet2020recalibrating, Google,zheng2020mitigating}. For instance, in the CAISO grid region, by shifting the controllable demands of data centers, up to 62\% of the yearly curtailment can be absorbed and up to 239 KtCO2 emissions can be reduced \cite{zheng2020mitigating}. Indeed, various appliances, devices and electricity demand have shown the potential of helping reduce carbon emissions associated with electric power, such as data center demand~\cite{radovanovic2022carbon, liu2011greening}, electric vehicle (EV)~\cite{cheng2022carbon}, and HVAC systems~\cite{10.1145/3600100.3626346}.

Among the efforts on reducing power system carbon emissions, existing works investigate the operational strategies to reduce carbon emissions. Such schemes are practical, as grid emission intensity is not uniform across the network, while load such as data center are geographically shiftable. Optimization and data-driven techniques such as linear programming (LP)~\cite{gonzato2021long}, mixed integer programming (MIP)~\cite{alizadeh2015dynamic}, reinforcement learning~\cite{jiao2023energy} are widely adopted. In the operation stage, the intra-day schedule strategies for various energy participants such as electric vehicles~\cite{jiao2023energy} and energy storage \cite{sun2023effective, he2021utility} are investigates. In the planning stage, strategies are developed to address (a) power generation expansion plans~\cite{alizadeh2015dynamic, fan2020integrated, 6397576, garcia2022priority}; (b) resource sizing and allocation~\cite{cheng2018bi, akbari2019optimal, gonzato2021long}; (c) operation scheduling~\cite{radovanovic2022carbon, Google, home2019optimal}; and (d) interplay between electricity, cooling, and heating loop~\cite{jiao2023energy}. 
Recently, the economy-wide decarbonization targets further require coordinated investment planning \cite{home2019optimal,fan2020integrated,lindberg2021guide, wang2023carbon,abdennadher2022carbon}. A classic instance is the transmission expansion planning (TEP) problem\cite{6397576} associated with investment restrictions, network flow equations, network capacity limits, generator limits, and device models. The TEP problem minimizes both capital expenditures (CAPEX) for transmission line investments and year-long operational expenses (OPEX), which can be tackled by a two-stage decision procedure for (a). here-and-now variables such as sizing-related decisions; and (b). wait-and-see operational variables including renewable curtailment and load shedding.
Similarly, coordinated investment planning problems aiming at mitigating generation-induced anthropogenic greenhouse gas emissions can be divided into two stages, and we refer to  Table \ref{LS: Literature review} for related works involving a variety of  temporal- and spatial-flexible resources.

However, such planning problems for carbon emission reduction are prohibitively expensive to solve due to large network sizes, various integer variables, complex network and device-level constraints, and long planning horizons. These challenges encourage us to investigate how to design an efficient solution strategies to reduce long-term system-level carbon emissions, which allow the planner to inspect possible solutions and identify the most promising ones.


In this paper, we focus on load-side planning and management to achieve the grid-level carbon reduction goal. Shiftable loads such as data center demands and EV fleet charging offer advantages of flexible scheduling in both time and space~\cite{zheng2020mitigating, lindberg2022using}. For instance, servers or data centers can be deployed across diverse locations, each with distinct carbon emission patterns associated with generation types, allowing for utilization of spatial variations in carbon intensity through load shifting \cite{lindberg2021guide}. Though normally perceived as an operational problem to coordinate load shifts throughout different locations, effectiveness of geographical load shifting can be further improved by appropriate placement of shiftable load locations. Well-planed placement of such shiftable loads enable the access to low-carbon electricity while also ensuring compliance with transmission limits~\cite{wang2023carbon,abdennadher2022carbon}, and it will help reduce long-term carbon emissions of the whole grid.

\begin{table*}[t] 
\vspace{-0.4cm}
\hspace{0.1cm}
\centering
\caption{{summary for the planning literature}} \label{LS: Literature review}
\setlength{\tabcolsep}{0.6mm}{
\begin{threeparttable}
\begin{tabular}{ccccccccc} 
\hline 
 Ref. & Objective & Resources & Carbon & Horizon & RPS & $\text{Timestep}_{max}$ & $\text{Scale}_{max}$ & Method \\ \hline
\cite{gonzato2021long} & CAPEX+OPEX & RES, ES   & $\times$   &  1-year    &   $\surd$     &  8760     &  NEM      &  LP                \\ 
\cite{alizadeh2015dynamic} &  CAPEX+OPEX  &   Transmission lines, fuel units  & $\times$ &  20-year & $\surd$  &  20   &    30-bus        & MIP+Benders decomposition \\
\cite{6397576} &  CAPEX+OPEX  &   Transmission lines   & $\times$ &  NEM & $\times$  &  3   &    4392-line       & MIP+B\&B+Search space reduction \\
\cite{home2019optimal} & CAPEX+OPEX & RES, ES, Gas turbines      & $\surd$   &  2-year    &   $\surd$     &  48     &  69-bus      &  MIP+ solver               \\ 
\cite{fan2020integrated} & CAPEX+OPEX+Emission & RES, EV      & $\surd$   &  3-year    &   $\surd$     &  864     &  54-bus      &  MIP+ solver                \\ 
\cite{wang2023carbon} & CAPEX+OPEX + Emission & RES, Data center   & $\surd$   &  3-year    &   $\surd$     &  NEM    &  30-node      &  MIP + solver                \\
\cite{abdennadher2022carbon} & Emission & RES, Data center   & $\surd$   &  1-year    &   $\surd$     &  24    &  73-bus      &  MIBLP + solver                \\
\textbf{Ours} & Emission & RES, Data center   & $\surd$   &  1-year, 20-year    &   $\Box$     &  8760    &  1888-bus      &  MIP+MCTS                \\
\hline
\end{tabular}
\begin{tablenotes}
\footnotesize
\item 1.This paper considers the CAPEX and OPEX in the constraints of planning model.~~~2.RPS: Representative periods.\\3.$\Box$: Available for both cases with or without RPS.~~~4.NEM: Not explicitly mentioned.
\end{tablenotes}
\end{threeparttable}
}
\vspace{-0.5em}
\end{table*}



To this end, we design a load shifting framework centered on a planning model to determine the sitting and operation strategies of such flexible loads. Given that the decision variables for the locations are integers, such planning problem is formulated as a Mixed Integer Programming (MIP) model. Standard decomposition-based approaches or some relaxation-based approaches to avoid integer programming have been proposed to find solution for large-scale MIP formulations. Branch and Bound (B\&B) algorithm is the most commonly used approach that systematically explores candidate solutions in a state space search \cite{clausen1999branch}. It relies on efficient estimation of lower and upper bounds for search space regions/branches. Without such bounds, the algorithm resorts to exhaustive search. Commercial solvers such as Gurobi, MOSEK and GLPK\_MI use sophisticated heuristics to direct the search process.  Long-term planning (LTP) for load shifting locations is necessary to align the planning of geographically shiftable loads with future system operation patterns \cite{LTP}. Standard approaches may fail under this situation due to numerous variables or fine-grained simulation steps \cite{brenner2023learning}. 
Existing works often select typical daily profiles as representative periods (RPS) for simulations to avoid such issues, and may result in suboptimal investment decisions \cite{home2019optimal,fan2020integrated, gonzato2021long, garcia2022priority, he2021utility,abdennadher2022carbon}. Thus, rather than scaling down the LTP, prioritizing solution method efficiency is advisable. Recemt advancements have been made in developing learning-based methods \cite{xavier2021learning,9964136,nair2020solving} to reduce search space and guide the variables selection for B\&B method. However, these techniques involve various hyperparameters, while relationship between hyperparameters and actual performance is hard to characterize. 


We propose a novel solution technique based on Monte Carlo Tree Search (MCTS) to efficiently construct and explore the state space of shiftlable load's locations and values. MCTS has been used for efficient search of decisions with a long and sequential pattern~\cite{sun2022symbolic,abe2022anytime,Khalil_2022}. In the tree representation, each edge denotes a move to assign a specific value to a chosen decision variable while each node represents an integral solution state. MCTS can identify the best direction to explore the candidate solution on accumulated statistics, which typically vary with the visit count and the average reward of nodes. Such statistics balance unvisited search area exploration and promising move exploitation, finally focusing on more promising regions, thus sufficiently reducing the search space~\cite{wang2020learning}. In addition,  B\&B method relaxes integer constraints during branch checking and can be time-consuming to reach a feasible integral solution.  MCTS can consistently sample integral solutions randomly for simulations. Consequently, MCTS can not only obtain a close-to-optimal solutions given a sufficiently large time, but also an operator can stop MCTS anytime to obtain a reasonable solution for validation and simulation studies \cite{abe2022anytime}. This is crucial due to the limited time window for solving long-term planning problems, as planners can swiftly assess decisions involving various planning variable combinations. Moreover, MCTS's performance is less impacted by instance-specific hyperparameters, enhancing its generalization ability and eliminating the need for tedious fine-tuning. 
Proposed MCTS is utilized to harness the flexible representation of search space with customized computational tools to guide the search tree expansion, which enables the large-scale planning model tractable to solve.
The main contributions of our work can be summarized as follows:
\begin{itemize}
\item [1)]
  A general carbon-aware load shifting framework is designed for controllable demand, which improves utilization of generation with low carbon intensity. This is achieved by shifting the load instead of solely transferring generation \cite{home2019optimal, fan2020integrated} or overlooking line limits \cite{zheng2020mitigating}.
  \item [2)]
  An MIP problem is formulated to minimize carbon emission by considering the CAPEX and OPEX and limits. The locations of controllable demand are determined for various time horizon, which involves long-term planning and fine-grained resolutions.
  \item [3)]
  An MCTS-based approach is adapted to solve the large-scale MIP problems, which can find reasonable solutions for the cases where basic B\&B, GLPK\_MI, Mosek, and Gurobi fail, and allow flexible time window settings along with offline solution adjustments.
  \item [4)]
  Numerical experiments are conducted across on systems up to 1888-bus system with 7,300 timesteps. These experiments validate the effectiveness of the proposed framework and the MCTS-based approach.
\end{itemize}






\section{Planning Model for LS Location}
\subsection{Controllable Demand for Carbon Reduction}
In the proposed framework, we consider the power network setup, where the nodes can have fuel and renewable generation, normal loads, and geographically shiftable loads, i.e., controllable demand of EV charging and data centers. We assume that the carbon emission arises from fuel generation, the impact of renewable curtailment, and demand shifting. For different nodes, the carbon intensity can be different. Thus, shifting demand can potentially change the total carbon emission. An illustrative example is shown in Table \ref{LS: Description} and Fig. \ref{LS:frame}. There are 4 nodes $A, B, C,$ and $D$ selected to conduct the demand shifting. By mitigating the controllable demand from A and C to B and D, the total carbon emissions and curtailment can be reduced. 

\begin{table}[h] 
\hspace{0.1cm}
\centering
\caption{Description for the nodes in the illustrative sample.} \label{LS: Description}
\setlength{\tabcolsep}{1.5mm}{
\begin{tabular}{cccccc} 
\hline
  & Generation & Data Center & Emission & Curtailment & Load shifting \\ \hline
A & Fuel        &   $\times$           & High            &   $\times$          & \multirow{2}{*}{A $\Rightarrow$ B} \\
B & Fuel \& RES &    $\times$        & Low             & $\surd$  &                   \\ \hline
C & Fuel        &    $\times$        & High            &    $\times$         & \multirow{2}{*}{C $\Rightarrow$ B\&D} \\
D & Fuel \& RES &    $\surd$          & Low             & $\surd$  &                   \\ \hline
\end{tabular}
}
 \vspace{-0.9em}
\end{table}

\begin{figure}[h] 
    \hspace{0.85cm}
	\centering
	\includegraphics[width=0.92\linewidth]{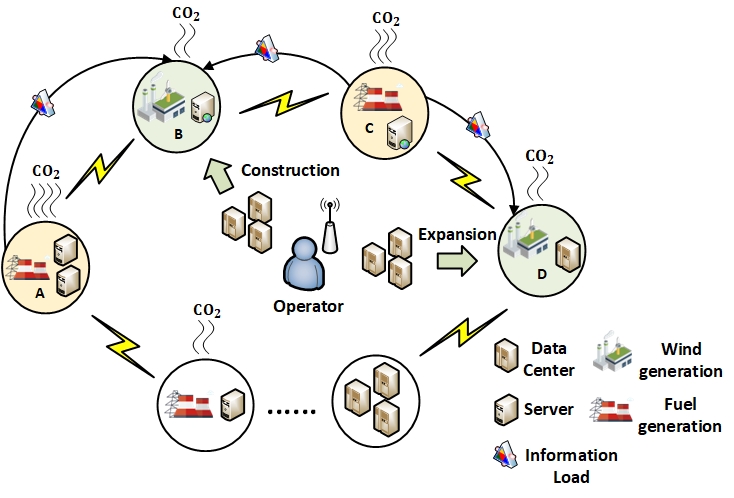}
   \vspace{-0.25cm}
	\caption{\footnotesize The schematic of our proposed spatial demand shifting framework. The shifting will be conducted among the selected locations to minimize the carbon emission of the whole system.}
	\label{LS:frame}
\end{figure}
Note that in this example, B and D need to absorb the extra loads, and intuitively the shiftable loads are supposed to be constructed in B and expanded in D under the limitation of the investment budget. More practical realization for 14-bus system is illustrated in Fig. \ref{LS:frame_14}, where 3 buses are selected to conduct load shifting. Generally, the upper-bound number of selected nodes can be set by the operator, and the optimal sitting and number of the selected nodes as well as the size of shifted controllable load will be determined by a long-term planning optimization problem.
\begin{figure}[t] 
    \hspace{0.25cm}
	\centering
	\includegraphics[width=0.8\linewidth]{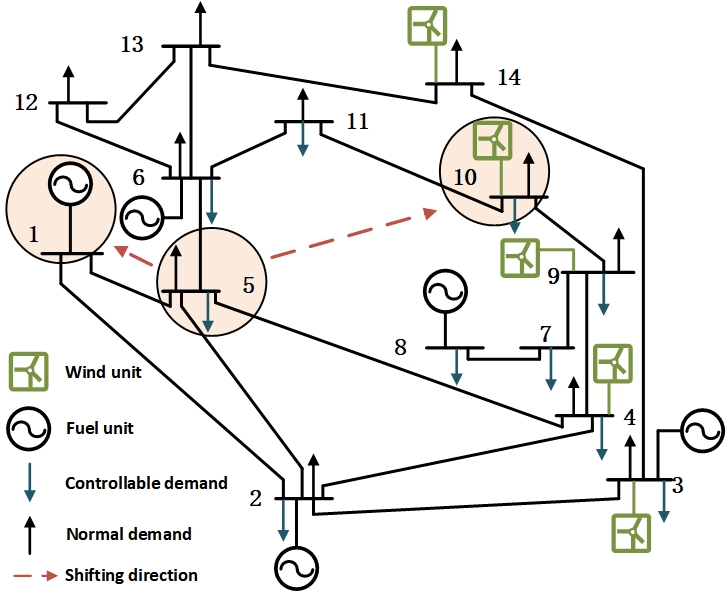}
	\caption{\footnotesize Illustration of IEEE 14-bus system and load shifting across regions. 
    Renewable generation such as wind generation, with a assumed carbon intensity of 0, is prioritized  for minimizing system-level carbon emissions. Bus 5 has the highest carbon intensity of generation consumption, while buses 1 and 10 have lower carbon intensities or surplus wind generation, and direct power provision is constrained by the line limits. Thus, the controllable load of bus 5 is shifted to buses 1 and 10 to optimize carbon emissions. }
	\label{LS:frame_14}
\end{figure}

\subsection{Objective function}
In this work, we focus on a planning optimization problem which aims at minimizing carbon emissions. It requires to consider emission sources across the network. For fuel generation,  we assume there is a quadratic relationship between the generation and carbon emission:
\begin{align}\label{LS: OBJ_FUEL}
C^{fuel}_{t}(\boldsymbol{p}^\text{gen}_{t}) = \sum^{g}_{i=1} [a_i(p^{\text{gen}}_{i,t})^2 + b_ip^{\text{gen}}_{i,t} ];
\end{align}
where $\boldsymbol{p}^\text{gen}_{t} \in \mathbb{R}^{g}$ denotes the generation at the buses. $g$ denotes the number of generators. $a_i$ and $b_i$ are the emission coefficients. Similarly, we describe the emission incurred by the demand variation as follows,
\begin{align}\label{LS: OBJ_load}
 C^{s}_{t}(\mathbf{\Delta}\boldsymbol{\ell}_{t}) = \sum^{n}_{i=1}[c_i({\Delta}{\ell}_{i,t})^2 + d_i{{\Delta}{\ell}_{i,t}}];
\end{align}
where $\mathbf{\Delta}\boldsymbol{\ell}_{t} \in \mathbb{R}^{n-1}$ denotes the shiftable load. $n$ denotes the number of buses. $c_i$ and $d_i$ are the emission coefficients.

The carbon emission from renewable curtailment, assumed to have a linear relationship, is considered as
\begin{align}\label{LS: OBJ_Cur}
C^{cur}_{t}(\boldsymbol{p}^\text{cur}_{t}) = \sum_{i=1}r^{q}_ip^{\text{cur}}_{i,t}
\end{align}
$\boldsymbol{p}^\text{cur}_{t} \in \mathbb{R}^{q}$ denotes the renewable curtailment. $r^{q}_i$ is the emission coefficient. 

Finally, the objective of the shiftable load planning model is formulated as
\vspace{-0.6em}
\label{LS: OBJ}
\begin{align}
\min_{\mathbf{z},\boldsymbol{p}^\text{gen},\mathbf{f},\boldsymbol{p}^\text{cur},\mathbf{\Delta}\boldsymbol{\ell}}\sum_{t=1}^{T} (C_{t}^{fuel}(\boldsymbol{p}^\text{gen}_{t}) + C_{t}^{cur}(\boldsymbol{p}^\text{cur}_{t}) + C_{t}^{s}(\mathbf{\Delta}\boldsymbol{\ell}_{t}))\Delta{t};
\end{align}
where $\mathbf{f}$ denotes the line flow described in the next subsection. $\mathbf{z} \in \mathbb{Z}^{n}$ denotes the decision variables for shiftable load location, where $z_i=1$ means such load is located at node $i$. $T$ is the total optimization horizon.
\subsection{System constraints}
In our planning problem formulation, the transmission network topology along with line parameters are included with linearized DC power flow models:

\begin{subequations}
\label{LS:Con_Sys}
\begin{align}
&\mathbf{A}\boldsymbol{p}^\text{gen}_{t} + \mathbf{B}\boldsymbol{p}^\text{res}_{t} + \mathbf{C}\mathbf{f}_{t} - \mathbf{B}\boldsymbol{p}^\text{cur}_{t} = \boldsymbol{\ell}_{t} + \mathbf{D}\boldsymbol{s}_{t} + \mathbf{M}\mathbf{\Delta}\boldsymbol{\ell}_{t}; \label{LS:Con_balance}\\
-&\overline{\mathbf{f}} \leq \mathbf{Kf}_t \leq \overline{\mathbf{f}};
\label{LS:Con_line} \\ &\mathbf{0}\leq \boldsymbol{p}^\text{gen}_{t} \leq \overline{\boldsymbol{p}}^\text{gen}; \label{LS:Con_gen}
\end{align}
\end{subequations}
where $\mathbf{f}_t \in \mathbb{R}^{n-1}$ denotes the fundamental flow. $\boldsymbol{p}^\text{res}_t \in \mathbb{R}^{n}$ denotes the renewable generation. $\boldsymbol{\ell}_t \in \mathbb{R}^{n}$ denotes the normal load. $\boldsymbol{s}_t \in \mathbb{R}^{n}$ denotes the original shiftable load. $\mathbf{A} \in \mathbb{R}^{n \times q}$ denotes the mapping from the fuel generations to the buses. $\mathbf{B} \in \mathbb{R}^{n \times g}$ denotes the mapping from the renewable generations to the buses. $\mathbf{C} \in \mathbb{R}^{n \times n-1}$ denotes the mapping from the fundamental flows to the buses.  $\mathbf{D} \in \mathbb{R}^{n \times d}$ denotes the mapping from the controllable loads to the buses. $\mathbf{K} \in \mathbb{R}^{m \times n-1}$ denotes the mapping from the fundamental flow to all flows. We use $\mathbf{M} \in \mathbb{R}^{n \times n-1}$ to represent the mapping from the fundamental shiftable load to all shiftable loads, which absorbs the shifting balance constraint that the summation of geographically shifted demand required to be 0. For instance, in a 3-node system, one available cases for $\mathbf{M}$ can be $[1,0;-1,-1;0,1]$. \eqref{LS:Con_balance} denotes the nodal power balance constraints. \eqref{LS:Con_line} denotes the line flow limitation. \eqref{LS:Con_line} denotes the fuel generation limitation. 

  \renewcommand{\algorithmicrequire}{\textbf{Input:}}
  \renewcommand{\algorithmicensure}{\textbf{Output:}}
  \begin{algorithm}[t]
    \caption{MCTS-based Shiftable Load Location}
    \begin{algorithmic}[1] \label{Algorithm: MCTS}
      \REQUIRE Planning model \eqref{LS: MODEL}, eploration rate $\rho$, round budget $K_{max}$, $\mathcal{T}_{all}=$ IPT.
      \ENSURE Location decision $\boldsymbol{z}^*$, decision tree $\mathcal{T}_{MCTS}$.
      \renewcommand{\algorithmicensure}{\textbf{Initialize:}}
      \ENSURE Best goal $\textit{L}^*=$ Inf, round $k = 1$, $N_r=0, V_r=0$.
      \WHILE {$\boldsymbol{z}^{(*)}$ doesn't converge $\mathbf{and}$ $k<K_{max}$}
      \WHILE{current node $i$ is not an expandable node ($N_i \leq 1$) or a leaf node}
        \STATE \textit{\textbf{Selection}:} go to next node $j$ that has maximal UCB among the children of node $i$.
      \ENDWHILE
      \STATE \textit{\textbf{Expansion:}}
      \\ If $N_i==0$: Go to step 6.
      \\ If $N_i==1$:
      \\a. Add all children of node $i$ to $\mathcal{T}_{MCTS}$.
      \\b. Let each child's $N_j=0, V_j=0$.
      \\c. Choose a child randomly to visit as current node.
      \STATE \textit{\textbf{Simulation:}}\\
      a. Start from the current node $i$ and randomly choose next action until reaching the terminal state $\boldsymbol{z}_k$.\\
      b. $\boldsymbol{z}_k \rightarrow$ \eqref{LS: MODEL} $\Rightarrow$ \eqref{LS: MODEL_LP}.\\
      c. Solve \eqref{LS: MODEL_LP} for each $t$ in parallel and get $QP^{*}_{t}$.\\
      d. Take $\sum{QP^{*}_{t}}$ as the reward $\gamma_k$.\\
      e. If $\gamma_k < \textit{L}^*$: $\boldsymbol{z}^*\leftarrow \boldsymbol{z}_k$.
      \STATE \textit{\textbf{Backpropagation:}} \\
      Start from the reached node, for each node $i$ have been visited in this round:
      $V_i \leftarrow V_i + \gamma_k, N_i \leftarrow N_i +1$.
\STATE $k \leftarrow k+1$.
      \ENDWHILE
      \STATE Return $\boldsymbol{z}^*, \mathcal{T}_{MCTS}$.
    \end{algorithmic}
  \end{algorithm}
\vspace{-1em}

\subsection{Load shifting constraints}
The shiftable load involves the curtailment of renewable generation and data center demand, which are bounded by \eqref{Con_curtail} and \eqref{Con_ls} respectively.
\begin{subequations}
\label{LS:Con_LS}
\begin{align}
&\mathbf{0}\leq \boldsymbol{p}^\text{cur}_t \leq \overline{\boldsymbol{p}}^\text{cur}; \label{Con_curtail} \\
&-diag\{\mathbf{z}\}\underline{\mathbf{\boldsymbol{\ell}}}_t\leq \mathbf{M}\mathbf{\Delta}\boldsymbol{\ell}_t \leq diag\{\mathbf{z}\} \overline{\mathbf{\boldsymbol{\ell}}}_t;
\label{Con_ls}
\end{align}
\end{subequations}
where \eqref{Con_curtail} and \eqref{Con_ls} limit the amount of renewable curtailment and shifted controllable demand, respectively. In \eqref{Con_ls}, the left-hand means the $\mathbf{\Delta}\boldsymbol{\ell}_t$ can be negative, i.e., shifted to other nodes, while the right-hand limits the amount of received controllable demand. Here we set the lower bound $\underline{\mathbf{\boldsymbol{\ell}}}_t = D\boldsymbol{s}_t$ and the upper bound $\overline{\mathbf{\boldsymbol{\ell}}}_t$ is determined by the basic and expanded size of the shiftable load. Note that only the controllable demand is limited by the integer variables $\mathbf{z}$, which requires the planning decisions to set the locations for such shiftable load. 

\subsection{Economic constraints} 

For our planning problem, we consider the following economic limits in the problem formulation:
\begin{subequations}\label{LS:Con_Eco}

\begin{align}
&C_{gen,t}(\boldsymbol{p}^\text{gen}_{t}) =\mathbf{w}^T\boldsymbol{p}^\text{gen}_t\leq \overline{C}_{gen,t}; \label{LS:Con_Eco_cost}
\\ &\sum_{i=1}^{n}z_{i} \leq K; \label{LS:Con_Eco_K}
\\ &\mathbf{\alpha}^{T}\mathbf{z} \leq B.\label{LS:Con_Eco_B}
\end{align}
\end{subequations}

In the above formulation, $\overline{C}_{gen, t}$ is the cap of fuel generation cost. $\mathbf{w}$ is the fuel generation cost coefficient. $K$ denotes the maximum number of selected locations. $\mathbf{\alpha}$ is the cost coefficient  associated with each locational choice of future shiftable load,  and $B$ is the maximum investment budget.

\subsection{Long-term planning optimization
model}
Finally, the planning model for LS can be summarized as follows,
\begin{subequations} \label{LS: MODEL}
\vspace{-0.6em}
\begin{align}
\min_{\mathbf{z},\boldsymbol{p}^\text{gen},\mathbf{f},\boldsymbol{p}^\text{cur},\mathbf{\Delta}\boldsymbol{\ell}} & \sum_{t=1}^{T} (C_{t}^{fuel}(\boldsymbol{p}^\text{gen}_{t}) + C_{t}^{cur}(\boldsymbol{p}^\text{cur}_{t}) + C_{t}^{s}(\mathbf{\Delta}\boldsymbol{\ell}_{t}))\Delta{t}\\
     s.t. \quad & \eqref{LS:Con_Sys} -\eqref{LS:Con_Eco}
\end{align}
\end{subequations}
where \eqref{LS: MODEL} model can be NP-hard due to its mixed
integer programming (MIP) nature, especially when applied to large-scale systems with a fine-grained time horizon, such as day-to-day granularity. This motivates us to design customized and more efficient solution scheme to reduce system emissions.
\vspace{-0em}

\section{MCTS Solution to Planning model}
\begin{figure*}[th] 
    \vspace{-0.15cm}
    \hspace{-0.cm}
	\centering
	\includegraphics[width=0.9\linewidth]{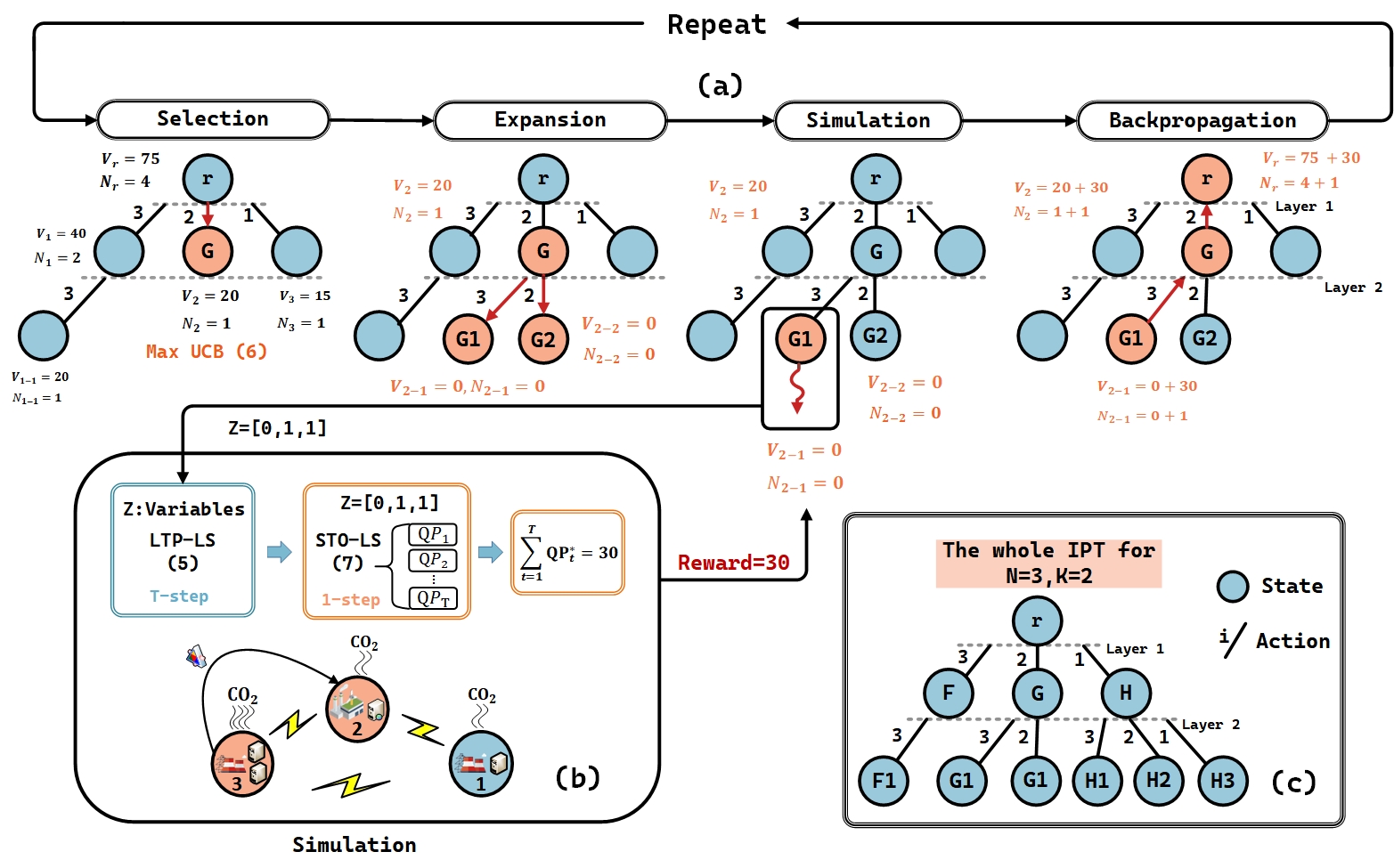}
	\caption{\footnotesize Illustrative example of MCTS-based location for a 3-bus system. We demonstrate the process of the 5th search round. (a) explains the 4 MCTS phase of one learning round. (b) show the specific sampling and reward calculation process in the simulation stage. (c) gives the tree structure of the IPT. Based on the given assumptions for $N_i$ and $V_i$, node G is determined to have the highest UCB value, and therefore, it is selected as the next node to visit. In this round, node G is visited but not expanded, resulting in the addition of its two child nodes, G1 and G2, to $\mathcal{T}_{MCTS}$. These child nodes are initialized with '$N_{G1}=0, V_{G1}=0, N_{G2}=0, V_{G2}=0$'. The action 'located at bus 1' is randomly chosen, leading to the selection of node G1 for the simulation stage. However, since there are no available actions at node G1, the simulation terminates at the state where the location decision is $\boldsymbol{z}$ = [1,1,0]. This indicates that bus 1 and bus 2 are selected for load shifting. \eqref{LS: MODEL_LP} can be solved in parallel to reduce the solution time. The reward for node G1 in this round is assumed as $\gamma_k = 30$, which is set as the sum of optimal objectives of \eqref{LS: MODEL_LP}. The UCB values for $\mathcal{T}_{MCTS}$ are updated through backpropagation: $N_{G1}=1, V_{G1}=30, N_{G}=2, V_{G}=50, N_{r}=105, V_{r}=5$. UCB values for other nodes remain the same. Given a sufficient time, $\mathcal{T}_{MCTS}$ can involve the most promising region of $\mathcal{T}_{IPT}$ and there will be a converged path start from the root to a leaf node, which can get a close-to-optimal objective.}  
	\label{LS:MCTS}
\end{figure*}


Traditional optimization algorithms address computation issues in solving \eqref{LS: MODEL} by aggregating nodes or using representative periods, leading to suboptimal decisions. To overcome these limitations, we propose a computationally efficient solution that avoids information reduction in this section. Especially, we utilize the iterative priority tree (IPT) and the Upper Confidence Tree (UCT)-based MCTS algorithm to determine the binary variables $\mathbf{z}$. 

\subsection{Tree Representation of Location Space}

We denote the set of possible solutions for $\mathbf{z}$ as $\Theta$. Then, we consider and denote a tree $\mathcal{T}_{all}$ where each node can represent a subset $\Theta_\mathbf{z} \in \Theta$, including several $\mathbf{z}$ patterns whose determined locations are partly same, and each leaf is a $\mathbf{z}$ pattern satisfying the limits \eqref{LS:Con_Eco_K} and \eqref{LS:Con_Eco_B}. Considering the following metrics \cite{abe2022anytime}:
\begin{enumerate}
\item Faithfulness: The inclusion relationship of the original location vector space remain unchanged: if node $j$ is a descendant of node $i$, then the locations selected in $\mathbf{z}(i)$ will be included in $\mathbf{z}(j)$, i.e., all the features of the vector $\mathbf{z}(j)-\mathbf{z}(i)$ are non-negative.
\item Non-redundancy: The mapping from a tree node into an location vector $\mathbf{z}$ is one-to-one. For any location vector $\mathbf{z} \in \Theta$ enabling \eqref{LS:Con_Eco_K} and \eqref{LS:Con_Eco_B} hold, there exists only one leaf of the constructed tree;
\end{enumerate}
We choose the IPT as illustrated in Fig. \ref{LS:MCTS} to represent the location space to ensure the correctness and improve the efficiency of the tree search. In the IPT, each node represents a state of current $\mathbf{z}$ and each edge corresponds to an action allocating a shiftable load location to one of the buses. The IPT selectively assigns a bus only if its priority is higher than that of the most prioritized bus which has already been assigned a shiftable load, where the priority can be predescribed by the operator. The proof of the faithfulness and non-redundancy of IPT is given in \cite{abe2022anytime}.

\subsection{Monte Carlo Tree Search}
In this subsection, we describe how to search iteratively in the constructed tree to find the solutions for the large-scale planning problem. To get an appropriate $\mathbf{z} \in \Theta$ from the constructed IPT, we employ the MCTS algorithm. This method effectively searches for optimal decisions within extensive search trees, such as the IPT representing the location space of large-scale systems. In this paper, we use the UCT to find a subtree $\mathcal{T}_{MCTS}$ from $\mathcal{T}_{all}$. This search involves a sequential process of tree traversal and we use $\mathcal{T}(k)$ to represent the subtree developed in traversal round $k$. Each round consists of the following four steps:

\textbf{(1)Selection.} From the root node, UCT traverses the nodes of current search tree $\mathcal{T}(k)$ and determines the next node based on a predefined selection policy, continuing this process until it reaches either an expandable node (go to expansion stage) or a leaf node (go to simulation stage). In this paper, we select the next node $j$ who is a child of the current node $i$ considering the value $v_j$ and the visited frequency of Node $j$. This can be represented as the Upper Confidence Bounds (UCB),
\begin{subequations} \label{UCB}
\begin{align}
UCB(j)& := \hat{v}_j + \rho\sqrt{\frac{log(N_i)}{N_{j}}}; \\
\hat{v}_j & := \frac{V_j}{N_{j}};
\end{align}    
\end{subequations}
where $\hat{v}_j$ is a estimator of $v_j$. $V_j$ is the sum of rewards and the reward will be calculated in the simulation step. $N_i$ denotes the number times node $i$ visited, and $N_{j}$ is the number of times node $j$ has been selected at node $i$. $\hat{v}_j$ encourages the exploitation of current best child node, while $\sqrt{\frac{log(N_i)}{N_{j}}}$  encourages the exploration of less-visited child nodes. $\rho \geq 0$ is a empirically predefined parameter, controlling the balance between exploring less visited nodes and exploiting nodes wit higher expected rewards. Node $j$ is selected because it has the maximal UCB among all the children of node $i$ and for a expandable node $i$, $N_i \leq 1$.  

\textbf{(2)Expansion.} When reaching an expandable node with $N_i = 0$, go to simulation stage directly. While reaching an visited node with $N_i = 1$, all its children are added to the tree with $N_j = 0$. The children are be determined according to \eqref{LS:Con_Eco_K} and \eqref{LS:Con_Eco_B}. Then, one of the children is randomly visited to conduct simulation.

\textbf{(3)Simulation.} For the reached node $i$, if it is non-terminal, one simulation starting from the current node until reaching the terminal state (e.g. [0,1,0] can be a terminal state for node G2 in Fig. \ref{LS:MCTS}) is conducted. The actions in this process are selected randomly. Once the terminal state is determined, the reward for the sampled vector $z_k$ can be calculated by input it to \eqref{LS: MODEL}.

\begin{proposition}
During the simulation stage, solving model \eqref{LS: MODEL} can be substituted with solving multiple smaller-scale submodels in parallel. In contrast, the basic B\&B method necessitates solving a relaxed model that closely matches the scale of \eqref{LS: MODEL} at each round.
\end{proposition}
\begin{proof}
Denote the model \eqref{LS: MODEL} with the sampled vector $z_k$ as $M_T$, then the feasible region of $M_T$ related to \eqref{Con_ls} can be fixed to the case of $z_k$. 

Consequently, for each single time step in $M_T$, we can formulate a submodel $M_t$ only involving the decision variables, constraints and objective function related to time step $t$, i.e., the feasible region and objective function of $M_t$ are independent of other time steps. Then solving $M_T$ will be equivalent to solve $M_t$ for $t=1,..,T$ which can be conducted in parallel. For the basic B\&B method, since $z$ at each round is still to be determined, the integer variables in the model \eqref{LS: MODEL} can be relaxed but the derived model can not be decoupled.
\end{proof}


Thus, here model \eqref{LS: MODEL} can be decoupled into $T$ independent and simple quadratic programming, which can be solved in parallel. While for methods like B\&B, it does not hold such properties for our planning problem, and is still computationally expensive with the heuristics. For each time step $t$, the quadratic programming (QP) can be formulated as:
\begin{subequations} \label{LS: MODEL_LP}
\begin{align}
QP_{t}^{*}:=\min_{\boldsymbol{p}^\text{gen}_t,\mathbf{f}_t,\boldsymbol{p}^\text{cur}_t,\mathbf{\Delta}\boldsymbol{\ell}_t} &  C_{t}^{fuel}(\boldsymbol{p}^\text{gen}_{t}) + C_{t}^{cur}(\boldsymbol{p}^\text{cur}_{t}) + C_{t}^{s}(\mathbf{\Delta}\boldsymbol{\ell}_{t})\\
& s.t. \quad \eqref{LS:Con_Sys} -\eqref{LS:Con_Eco_cost}\\
&~~~~~\quad\mathbf{z}_t = \mathbf{z}_{input}
\end{align}
\end{subequations}
The reward, denoted as $\gamma_k$ is defined as the negative sum of the objective values of the above quadratic programmings. While load shifting can enhance transmission network feasibility for more load patterns, \eqref{LS: MODEL_LP} may not always be feasible. For the infeasible case, we remove the corresponding terminate state, i.e., the leaf node in the $\mathcal{T}_{all}$ and sample another solution. If all the children of a node are removed, the node itself is also eliminated. In case the current reached node $i$ is removed, we return to the selection stage with the updated $\mathcal{T}_{all}$.

\textbf{(4)Backpropagation.} Start from the reached node, for each node $i$ have been visited in this round, we do this backpropagation: $V_i \leftarrow V_i + \gamma_k, N_i \leftarrow N_i +1$. As a result, the statistics of all nodes are updated, serving as the starting point for the next round.

The iteration of these four steps will stop either when the visiting path has converged or when the computation budget has been exhausted, and the solution's availability can be guaranteed considering the analysis below.

\begin{remark} 
Algorithm \ref{Algorithm: MCTS} is capable of producing feasible and reasonable solutions when stopping early.
\end{remark}
This is attributed to the sampling approach and the best-first search property of the MCTS, which collectively assure a solution's availability at any time and facilitate continuous result refinement. This ensures adaptability of the proposed Algorithm \ref{Algorithm: MCTS} to varying computational budgets during runtime, with solutions improving over time and converging to the optimum given sufficient time and memory.



\section{Case Study}
In this section, we evaluate the performance of our proposed long-term planning scheme and MCTS-based shiftable load location using real and synthetic datasets. we demonstrate the proposed scheme is carbon-effective while the solution procedure is efficient and scalable compared to Basic B\&B method and practical solvers (Gurobi, Mosek and GLPK\_MI). Specifically, our scheme achieves more than a 10\% reduction in emissions across all settings. MCTS-based  method provides an acceleration of up to 8.1X in cases with large time steps that are solvable for all methods, as long as the solution gaps are less than 1.5\%. Furthermore, our method is capable of finding a reasonable solution in cases where other methods fail.

\begin{table}[b]
\hspace{0.1cm}
\centering
\caption{Description for the nodes in the illustrative sample.} \label{LS: Description}
\vspace{-0.5em}
\setlength{\tabcolsep}{0.4mm}{
\begin{tabular}{cccccc}
\hline
         & No.Lines($m$) & No.Gens($q$) & No.RES($g$) & No. Load ($d$) & No. Integers\\ \hline
14-bus   & 20       & 5       & 5      & 10    & 14 \\
39-bus   & 46       & 10      & 5      & 30    & 39 \\
118-bus  & 186      & 54      & 25     & 80    & 118\\
1888-bus & 2308     & 285     & 500    & 1500  & 1888\\ \hline
\end{tabular}
}
\end{table}

\subsection{Simulation Setups}
For our analysis, we use IEEE 14-bus, IEEE 39-bus, IEEE 118-bus and IEEE-1888rte systems and consider the optimization horizons up to twenty years. The specific configurations of the investigated systems and the number of integer variables in each setting are given in Table \ref{LS: Description}. The typical daily profiles in 2000-2019 of wind generation $\boldsymbol{p}^{res}_t$, normal load $\boldsymbol{l}_t$ \cite{RTE} and controllable load $\boldsymbol{s}_t$ \cite{DCLoad} are shown in Fig. \ref{LS:Average_level} and Fig. \ref{LS:Average_level_hourly}. In Figure \ref{LS:Average_level}, we have a time horizon of twenty years, where each data point represents the peak level within a day. We utilize these profiles to conduct simulations with a daily resolution, using time steps of 30, 365, and 7300 days. For the one-year time horizon, we consider the dataset of 2011. In Figure \ref{LS:Average_level_hourly}, we randomly select the profiles of two adjacent days in 2011. Based on these profiles, we conduct simulations with an hourly resolution, using time steps of 720 and 8760 hours, respectively. The dataset of wind generation and normal load is in hourly intervals and is published by RTE, which is the France's Transmission System Operator. 

Due to the missing wind generation data from 2000-2010 in \cite{RTE}, we utilize existing data on wind power generation distribution to fill the gap. In \cite{DCLoad}, the hourly controllable load for a year is derived from real hourly data center demand observed over a two-week period. In our simulations, this data is utilized directly for a 1-year horizon and further projected by normal distribution to a twenty-year horizon with assuming a 15\% annual increasing rate. Then we allocate the normal loads to each node proportionally according to the standard cases in Matpower, while the controllable loads and the wind generations are assigned to a portion of buses with scaling the typical profiles. 


The parameter settings of MCTS based method and Gurobi optimizer for different cases can be found in Table 3-5. All experiments have been implemented using Python 3.8 in a Laptop with an Intel Core i5 processor at 2.5 GHz with 16GB RAM.  Gurobi is used as the QP problem solvers.

\begin{figure}[t] 
    \hspace{-0.3cm}
	\centering
	\includegraphics[width=0.95\linewidth]{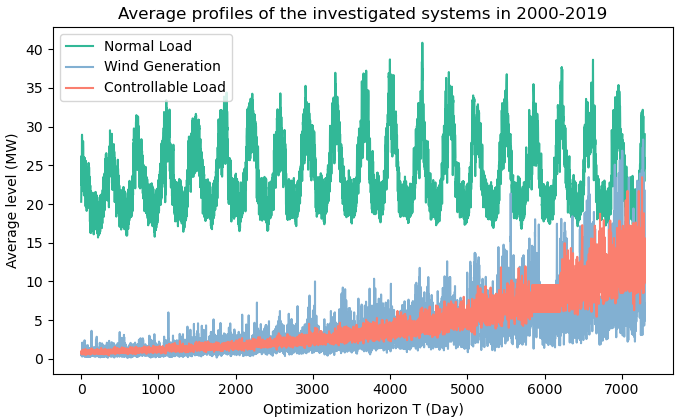}
	\caption{\footnotesize The average levels of the normal load, wind generation, and controllable load for per node in the investigated systems. The values represent the peak values within a day and the profiles are scaled to different levels according to the scales of systems.}
	\label{LS:Average_level}
 \vspace{-0.9em}
\end{figure}

\begin{figure}[t] 
    \hspace{-0.3cm}
	\centering
	\includegraphics[width=0.95\linewidth]{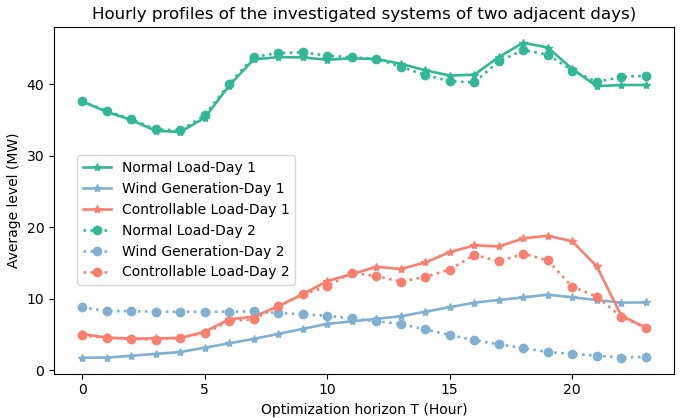}
	\caption{\footnotesize The hourly profiles within a day.}
	\label{LS:Average_level_hourly}
 \vspace{-0.9em}
\end{figure}




\subsection{Evaluation of Carbon Emission Reduction}
To evaluate the effect of the proposed framework in decarbonization, we compare the carbon emission with the DC-OPF model defined as follows:

\begin{subequations}
\vspace{-0.2em}
\label{LS:Metric_total}
\begin{align}
\min_{\boldsymbol{p}^\text{gen},\mathbf{f}}\quad & \sum_{t=1}^{T} C_{gen,t}(\boldsymbol{p}^\text{gen}_{t})\\
\text { s.t. } \quad & \eqref{LS:Con_balance}-\eqref{LS:Con_gen}
\end{align}
\end{subequations}
In the simulations, no shiftable loads are allowed. We adjust the cap of fuel generation cost ($\overline{C}{gen,t}$) to 1.05 times the optimal solution ($C{gen,t}(\boldsymbol{p}^{gen,*}_{t})$) in \eqref{LS:Con_Eco_cost}. The operator can modify this cap according to requirements.

\begin{table}[b]
\vspace{-0em}
\hspace{0.1cm}
\centering
\caption{Carbon emission in Long-term Planning.} \label{LS: LTP_Metric}
\vspace{-0.5em}
\setlength{\tabcolsep}{1.9mm}{
\begin{tabular}{cccclc}
\hline
         & $\mu_{incr}$ & $C_{OPF}$ & $C_{LS}$ & $\mu_{redu}$ & Round \\ \hline
14   & 5\%       & 10828789.62               & 8827543.46                   &   18\%         & 70     \\
39   & 5\%       & 12341884.52     &   10313847.81
      &      17\%      & 165       \\
118  & 5\%      & 87253403.27    &    75316296.72           &     14\%       &   30 \\
1888 & 5\%    & 1371929726.28   & 1186488872.36       &   13\%       & 10    \\ \hline
\end{tabular}
}
\end{table}

We firstly consider the hourly resolution cases with $T = 720$ and $T = 8760$ for 14-bus and 39-bus system, and the results of carbon mission derived by the DC-OPF model and the LS-model are shown in Fig. \ref{LS:LS_emission_hourly}. It can be seen that both systems can reduce more than 10\% carbon emission. 
\begin{figure}[b] 
   \hspace{0.1cm}
	\centering
	\includegraphics[width=0.98\linewidth]{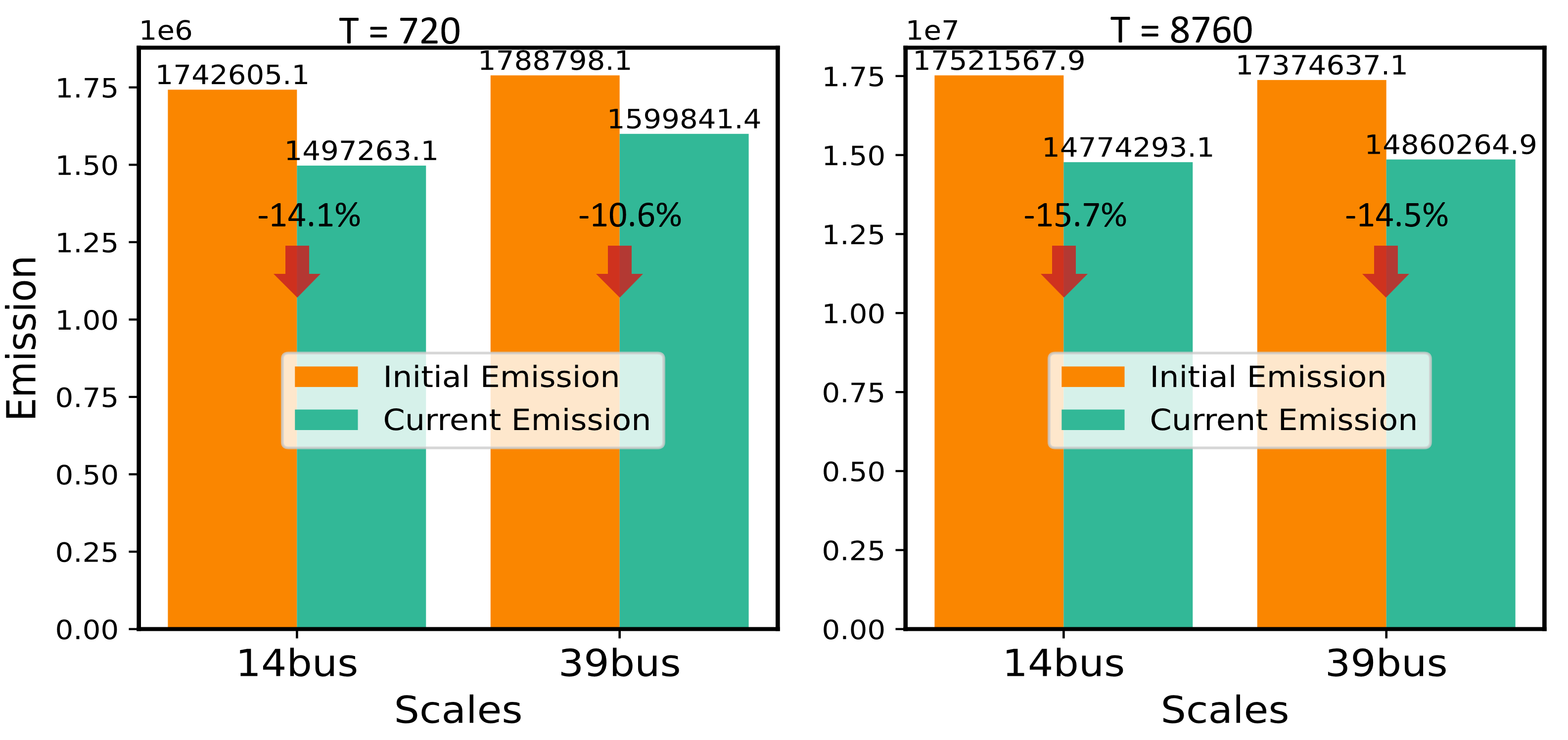}
   \vspace{-0.35cm}
	\caption{\footnotesize Carbon emission of 14-bus and 39-bus systems with the hourly resolution.}
	\label{LS:LS_emission_hourly}
\end{figure}

\begin{table*}[h]
\centering
\caption{Description for the nodes in the illustrative sample.} \label{LS: Metric_per}
\setlength{\tabcolsep}{2.4mm}{
\begin{tabular}{c|cccc|cccc|cccc}
\hline
\multirow{2}{*}{} & \multicolumn{4}{c|}{Metric($0.25 \Delta\overline{l}$)}                                                       & \multicolumn{4}{c|}{Metric($ \Delta\overline{l}$)}                                                       & \multicolumn{4}{c}{Metric($1.75 \Delta\overline{l}$)}                                                        \\ \cline{2-13} 
                  & $\mu_{redu}$ & $\mu_{allow}$ & $\mu_{shift}$ & $\Delta\hat{\boldsymbol{l}}$ & $\mu_{redu}$ & $\mu_{allow}$ & $\mu_{shift}$ & $\Delta\hat{\boldsymbol{l}}$ & $\mu_{redu}$ & $\mu_{allow}$ & $\mu_{shift}$ & $\Delta\hat{\boldsymbol{l}}$ \\ \hline
14           &      16\%          &      1.27         &    3.12          &   135009.89                  & 18\%         & 0.67          & 4.11           & 217287.89     &    19\%      &   0.42            &      4.14         &   235502.51                          \\
39                &   14\%           &     0.71          &      2.20         &  187968.87                & 17\%         & 0.39          & 2.18          & 424456.14      &      17\%        &      0.25         &      2.16         &    487856.14                          \\
118               &      13\%        &      0.49         &     1.48          &    1221542.34                        & 14\%         & 0.33          & 1.51          & 1416302.76     &    14\%          &    0.23          &     1.57          &      1616048.26                        \\ \hline
\end{tabular}
}
\end{table*}

Furthermore, solving 20-year planning problem can give more information to guide investment and construction of shiftable loads. Thus, we solve \eqref{LS: MODEL} and \eqref{LS:Metric_total} with daily resolution, i.e., $T=7,300$ for four size of systems using the profiles given in Fig. \ref{LS:Average_level}, and the results of carbon mission are given in Table \ref{LS: LTP_Metric}. It is expected the carbon emissions are positively correlated with system scales due to larger systems having more load consuming electricity. However, the situation regarding emission reduction is actually the opposite. The 14-bus system has the highest reduction of 18\%, while the 1888-bus system has the lowest reduction of 13\%. This could be attributed to a decrease in the percentage of shiftable load as the system scales increase or the small differences in carbon intensity. To illustrate this situation and analysis the effect of load shifting, we introduce the metrics reduction per allowed MW $\mu_{allow}$ and reduction per shifted MW $\mu_{shift}$ proposed in \cite{lindberg2021guide}, the definitions of them are given in \eqref{LS:Metric_single}:
\begin{align}\label{LS:Metric_single}
\mu_{allow} = \frac{\delta}{\mathbb{L}},\quad
\mu_{shift} = \frac{\delta}{\mathbb{S}};
\end{align}
where $\delta$ denotes the total emission reduction from minimizing the carbon emission without load shifting. $\mathbb{L} = \sum_{t \in \mathcal{T}}{\mathbf{\Delta}\overline{\mathbf{\boldsymbol{\ell}}}_t}$ denotes the maximum amount of shifted load, while $\mathbb{S} = \sum_{t \in \mathcal{T}}{|\mathbf{\Delta}{\mathbf{\boldsymbol{\ell}}}_t|}$ denotes the total amount of load shifted. The units of $\mu_{allow}$ and $\mu_{shift}$ are carbon tons per MW.
\begin{figure*}[th] 
    \hspace{-0.5cm}
	\centering
	\includegraphics[width=1.0\linewidth]{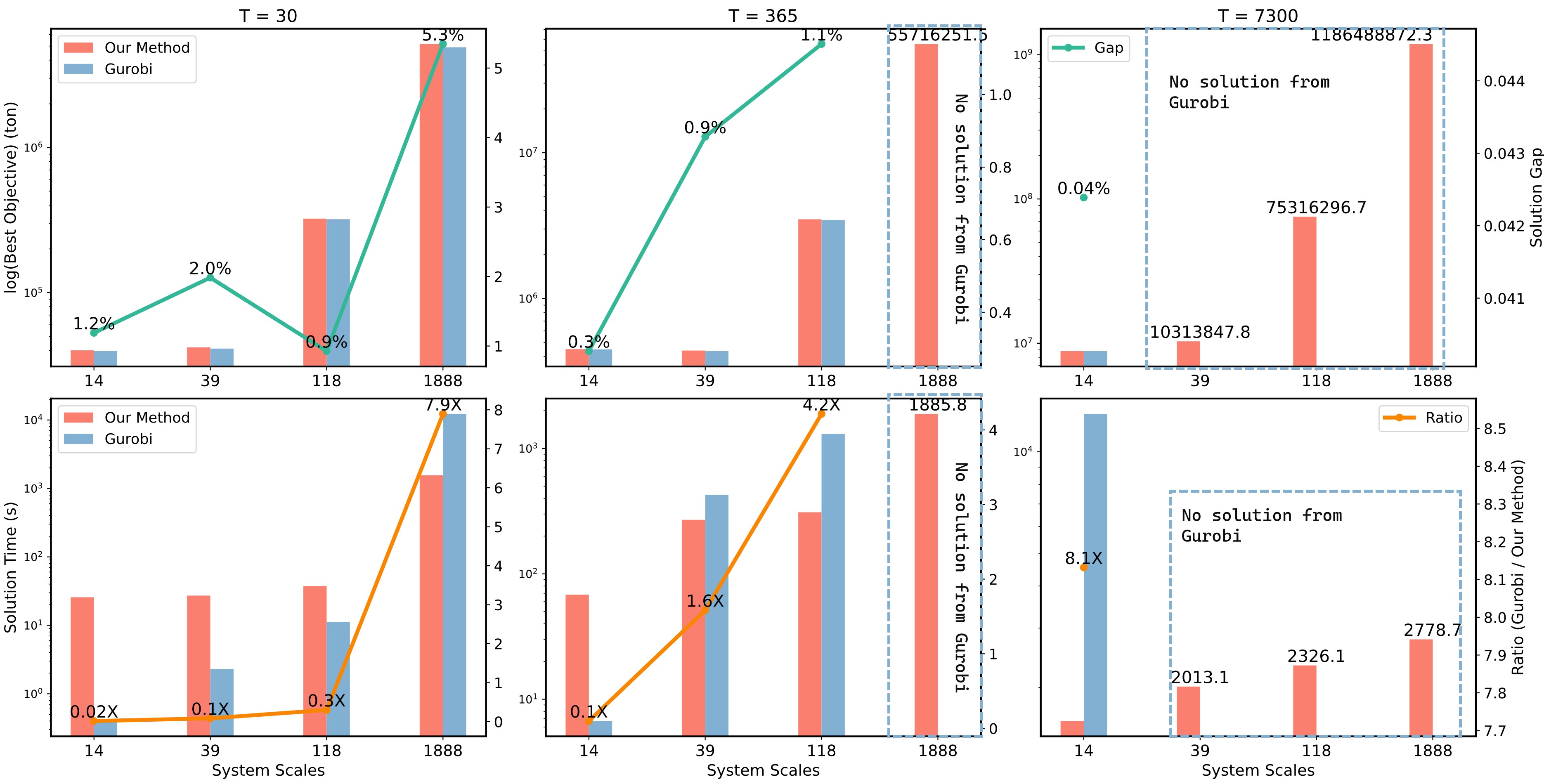}
	\caption{\footnotesize Performance of MCTS on different settings of scale and time horizon compared to Gurobi. Note Gurobi fails on $T=7,300$ cases when simulated on larger networks, while MCTS approach can always produces feasible plan.}
	\label{LS:MCTS_yearly}
\end{figure*}

Corresponding results for such metrics are presented in Table \ref{LS: Metric_per}. For all settings, the proposed planning model can still achieve more than 10\% carbon emission reduction. By changing the upper bound of shiftable load, we can alter the percentage of shifted load. Comparing the results of 14-bus and 39-bus systems, it can be observed that the variation in shifted load from $0.25\Delta\overline{\boldsymbol{l}}$ to $\Delta\overline{\boldsymbol{l}}$ is larger than the variation from $\Delta\overline{\boldsymbol{l}}$ to $1.75\Delta\overline{\boldsymbol{l}}$. Meanwhile, for 118-bus system, the variation is relatively steady. This may be because $\Delta\overline{\boldsymbol{l}}$ is large enough and the most effective carbon-efficient load of some buses have reached their limitations within the predefined upper bounds $\Delta\overline{\boldsymbol{l}}$ for the 14-bus and 39-bus systems. Thus, increasing the upper bound does not have a significant impact on the shifted load of these buses. Moreover, load shifting among other buses does not contribute significantly to carbon emissions. This leads to a surplus of shiftable load, resulting in slow or even decreased increases in $\mu_{allow}$ and $\mu_{shift}$. The gap between the $\mu_{allow}$ and $\mu_{shift}$ indicates that there is room for further improvement in utilizing the capacity of shiftable load. In practice, operators can allocate a larger size of shifted controllable load to carbon-efficient buses based on prior knowledge.



\subsection{Evaluation of MCTS-based Shiftable Load Location}
In the above simulations, it is observed that the effectiveness of Gurobi on a specific scenario largely relies on the suitability of its heuristics for that scenario. For example, in 39-bus system scenarios, Gurobi can find the optimal solution within 6 hours when $T=8760$. However, for a smaller problem scale with $T=7300$, Gurobi fails to find an available solution even after more than 10 hours and thus we turn to use the proposed MCTS-based method to find the solution.

To verify the performance of the MCTS-based Location method, we firstly compare the solutions given by the proposed method and Gurobi for four scale systems on three time step settings, the results are given in Fig. \ref{LS:MCTS_yearly}.
When $T=30$, Gurobi outperforms the proposed method in solving 14-bus, 39-bus, and 118-bus scenarios faster, with at least 2\% smaller carbon emission. For the 1888-bus scenario, Gurobi still finds the optimal solution and, with a time limit of under 2,000s, the proposed method finds a solution with a 5.3\% gap. With $T=365$, Gurobi solves the 14-bus, 39-bus, and 118-bus systems but fails for the 1888-bus system. The solution gaps for the 14-bus and 39-bus systems are 0.3\% and 0.9\% respectively, while it is 1.1\% for the 118-bus system. Despite being slower than Gurobi for the 14-bus cases, our method achieves 1.6X and 4.2X speed increases for the 39-bus and 118-bus systems respectively. For the 1888-bus system, our method finds a reasonable solution within 2000s. When $T=7,300$, Gurobi is able to find the optimal solution for the 14-bus system, but it fails to find feasible solutions for the 39-bus, 118-bus, and 1888-bus systems within 10 hours. In contrast, our method manages to find reasonable solutions for all cases within 3,000 seconds. Notably, for the 14-bus systems, the solution gap is only 0.04\%. These results verify that compared to Gurobi, our method delivers superior solution speed and feasibility for large scale MIP problems.
\begin{figure*}[t] 
   \vspace{-1.5em}
    \hspace{-0.6cm}
	\centering
	\includegraphics[width=1.03\linewidth]{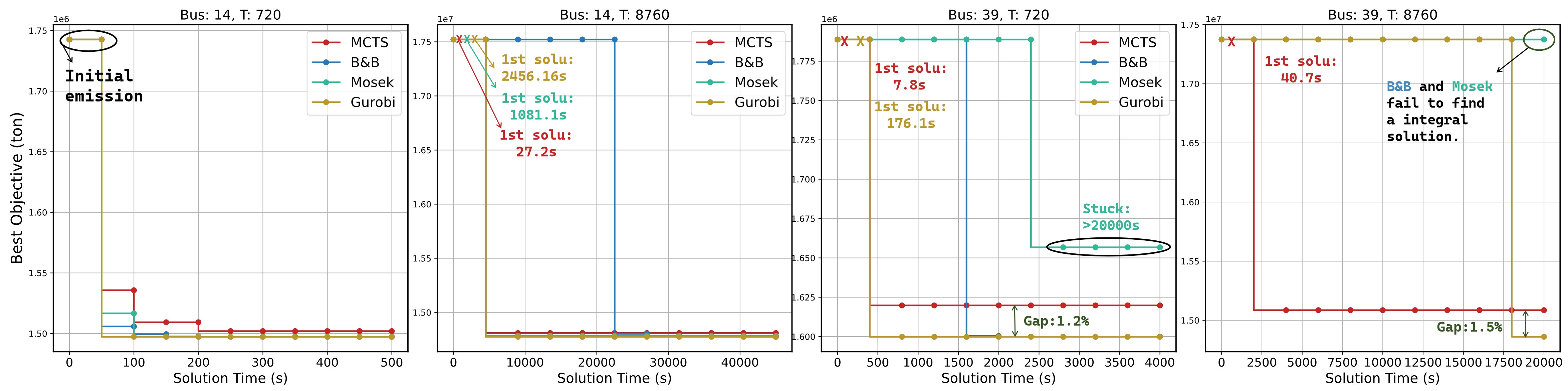}
	\caption{\footnotesize Flexibility analysis of solution time window.}
	\label{LS:MCTS_Flexibility}
\end{figure*}

Furthermore, we compare the solution process of our method with that of basic B\&B, Mosek, and Gurobi, in order to examine the configuration of solution time windows. The results are illustrated in Fig. \ref{LS:MCTS_Flexibility}, for 14-bus system with $T=720$, all solvers can find reasonable solution within 50s, and Gurobi achieve the least carbon emission. For 14-bus system with $T=8,760$, except the basic B\&B method, the other solvers can find reasonbale solutions within 5000s. Specifically, the time taken to find the first integral solution by our method, Mosek, and Gurobi are 27.2s, 1,081.1s, and 2,456.2s, respectively. Gurobi maintains the lowest level of carbon emissions, a benchmark the other solvers can closely match. Notably, the basic B\&B method can find an identical solution, but it requires a significantly longer time window of 41,400.1s.For a 39-bus system with T=720, both our method and Gurobi can find reasonable solutions within 500s. Our method needs 7.8s to find the first integral solution, while Gurobi takes 176.1s. The solution gap between our method and Gurobi is about 1.2\%. In contrast, the basic B\&B method takes more than 1600s to find an integral solution and over 3,500s to find the optimal solution. Mosek takes 2,591.3s to find its first integral solution, which is subpar, and then needs over 20,000s to find the next integral solution. For a 39-bus system with T=8,760, our method can find a reasonable solution within 2,500s, and the gap is 1.5\% compared to the solution given by Gurobi, which takes more than 17500s. However, both the basic B\&B and Mosek methods are unable to find an integral solution within 10 hours. The differing results among our method and other solvers demonstrate the efficiency of our approach, which can halt after a short time to deliver a reasonable solution. This discrepancy may be due to the solution search process: MCTS continues to evaluate the integral solution while other solvers opt to relax some integral variables. This reinforces the viability of our method to stop at any time.
\begin{figure*}[t] 
	\centering
	\includegraphics[width=1.0\linewidth]{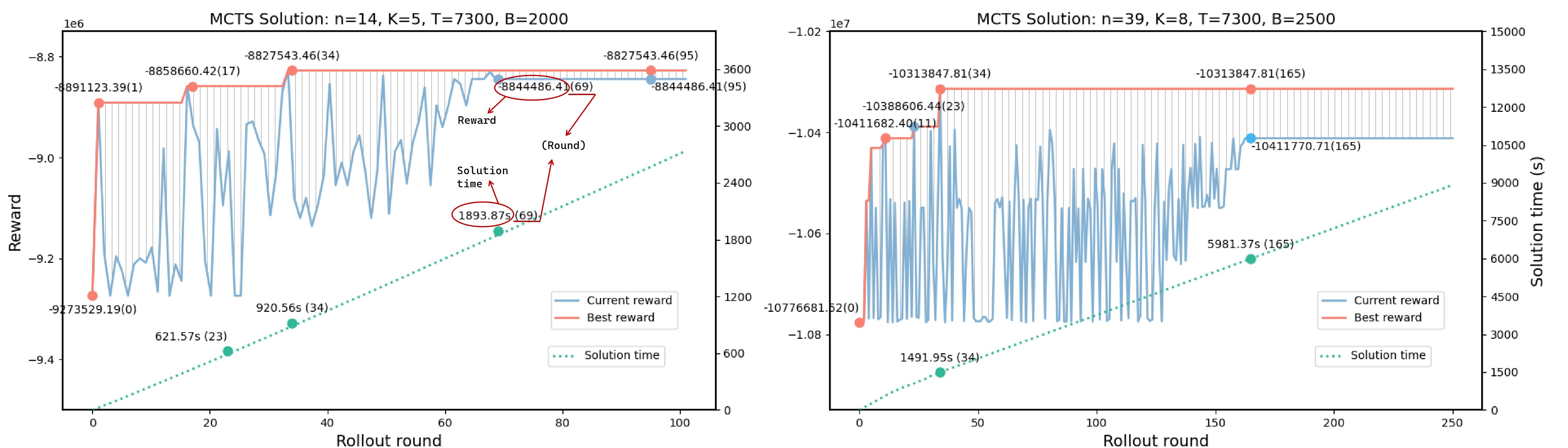}
 \caption{The specific variation in the solution search process of the MCTS method applied in 14-bus and 39-bus systems.}
	\label{LS:14_Convergence}
 \vspace{-1em}
\end{figure*}

The specific variation in the solution search process of MCTS method is presented in Fig. \ref{LS:14_Convergence}. It can be found that the current reward, i.e., the value of objective function keep changing in the search process, since in each round MCTS will evaluate a random solution by conducting the simulation mentioned in Fig. \ref{LS:MCTS}. Specifically, for the 14-bus system, the optimal reward — considered as the reasonable solution — is found at round 34, with the search process converging at round 69. For the 39-bus system, the reasonable solution is updated at round 34, and the process converges at round 165. Therefore, by setting time windows longer than 920.6s and 1491.9s for the 14-bus and 39-bus systems respectively, the MCTS method can already obtain a reasonable solution. Note that since each simulation round solves the same scale and number of QP models, it takes similar time to update the reward. Therefore, the solution time is nearly proportional to the number of rounds. Moreover, in Table \ref{LS: Description}, it can be seen that there is no clear correlation between the value of $K$ and the solution time when $K$ reaches a certain threshold. This is likely because as $K$ increases, more infeasible combinations of locations are included, but they are not evaluated, leading to negligible impact on the solution time.

\begin{table}[]
\hspace{0.1cm}
\centering
\caption{Convergence results for 14-bus system with different K.} \label{LS: Description}
\setlength{\tabcolsep}{2.8mm}{
\begin{tabular}{ccccc}
\hline
K                     & Convergent Reward & Best Reward                    & Round & Time                        \\ \hline
2                     & -9166178.76       & -9028317.63                    & 30    & 723.16                      \\
3                     & -9054640.63       & -8896583.28                    & 37    & 900.03                     \\
4                     & -8928860.53       & -8870661.47                    & 50    & 1229.64                     \\
5                     & -8844486.41       & -8827543.46                    & 69    & 1893.87                     \\
6                     & -8893019.48       & -8891123.38                    & 50    & \multicolumn{1}{l}{1218.68} \\
\multicolumn{1}{l}{7} & -8830158.51       & \multicolumn{1}{l}{-8819644.28}           & 58    & \multicolumn{1}{l}{1396.14} \\
\multicolumn{1}{l}{8} & -8830158.51       & \multicolumn{1}{l}{-8823788.95} & 65    & \multicolumn{1}{l}{1604.93} \\
\multicolumn{1}{l}{9} & -8830158.51       & \multicolumn{1}{l}{-8815393.02} & 70    & \multicolumn{1}{l}{1804.93} \\ \hline
\end{tabular}
}
\end{table}

\begin{table}[]
\hspace{0.1cm}
\centering
\caption{Chosen Buses in OPtimal Path for 14-bus system (K=5).} \label{LS: Description_path}
\setlength{\tabcolsep}{4.3mm}{
\begin{tabular}{cccc}
\hline
Round & Path        & Reward       & Time (s)   \\ \hline
15    & [0]          & -8981614.735 & 388.4  \\
53    & [0, 4]       & -8969392.451 & 1340.1 \\
59    & [0, 4, 8]      & -8827543.464 & 1440.6 \\
63    & [0, 4, 8, 9] & -8844486.411 & 1541.2 \\
69    & [0, 4, 8, 9, 11] & -8844486.411 & 1893.8 \\ \hline
\end{tabular}
}
\end{table}

Table \ref{LS: Description_path} illustrates the step-by-step determination of the converged solution for 14-bus system with $K=5$. 
At round 15, bus 1 has the highest UCB and is selected for expansion in the search path. At round 53, bus 1 maintains the highest UCB in the first layer. Among all the children of bus 1, bus 5 has the highest UCB in the second layer, resulting in a path change to [0,4]. Bus 5 is then expanded. By the same reasoning, buses 1, 5, 9, 10, and 11 have the highest UCB in their respective layers. With a relatively small exploration rate, the increase in visited times has a minimal impact on the UCB. As a result, the path converges to [0,4,8,9,11] and the search space are reduced compared to explore the whole tree.






\section{Conclusion}
In this paper, we develop a novel planning framework for expanding and operating geographically shiftable loads. The proposed framework can reduce system-level carbon emissions by over 10\% and limit the rise in generation costs to less than 5\%. With a longer planning horizon, the determined locations are likely to result in better long-term decarbonization effects. To solve the long-term and large-scale planning model computationally, we propose to utilize an adapted MCTS algorithm to customize the solution procedure of the resulting mixed-integer problem. In all investigated cases with different time windows, our MCTS-based approach can find the promising search region containing reasonable solutions of shiftable load locations, which can reduce carbon emission sufficiently. In future work, we will also consider robust operation and corresponding market design for geographically shiftable resources.

\bibliographystyle{IEEEtran}
\bibliography{bib}

\end{document}